\documentclass[10pt, a4paper]{article}
\usepackage[a4paper,left=1.1in,top=1.2in,right=1.1in,bottom=1.3in]{geometry}

\usepackage[noadjust]{cite}
\usepackage[dvips]{graphicx}
\usepackage[caption=false]{subfig}
\usepackage{amsfonts,amsmath,amssymb,amsthm,mathrsfs}
\usepackage{xurl}
\usepackage{color,soul}
\usepackage{xcolor}
\usepackage{ragged2e}
\usepackage{changes}
\usepackage{comment}
\usepackage{pgfplots}
\pgfplotsset{compat=newest}
\usetikzlibrary{patterns}

\usepackage{algorithmic}
\def\BibTeX{{\rm B\kern-.05em{\sc i\kern-.025em b}\kern-.08em
    T\kern-.1667em\lower.7ex\hbox{E}\kern-.125emX}}

\DeclareMathOperator{\E}{\mathbb{E}}

\newcommand{\cS}{\mathcal S}
\newcommand{\cA}{\mathcal A}

\newcommand{\Magg}{M^{\text{\scriptsize ml}}}
\newcommand{\Sagg}{\mathcal S^{\text{\scriptsize ml}}}
\newcommand{\Aagg}{\mathcal A^{\text{\scriptsize ml}}}
\newcommand{\qagg}{q^{\text{\scriptsize ml}}}
\newcommand{\ragg}{r^{\text{\scriptsize ml}}}

\newcommand{\Psiagg}{\Psi^{\text{\scriptsize ml}}}

\newtheorem{lemma}{Lemma}

\newtheorem{assumption}{Assumption}

\newcommand{\footremember}[2]{%
    \footnote{#2}
    \newcounter{#1}
    \setcounter{#1}{\value{footnote}}%
}
\newcommand{\footrecall}[1]{%
    \footnotemark[\value{#1}]%
} 

\title{Scaling Power Management in Cloud Data Centers: A Multi-Level Continuous-Time MDP Approach}

\author{Behzad~Chitsaz\footremember{ECETehran}{School of Electrical and Computer
Engineering, College of Engineering, University of Tehran, Tehran, Iran.} \and 
        Ahmad~Khonsari\footrecall{ECETehran} 
        \footremember{IPM}{School of Computer Science, Institute for Research in Fundamental Sciences, Tehran, Iran.} \and Masoumeh~Moradian\footrecall{IPM} 
        \and Aresh Dadlani\footnote{Department of Electrical and Computer Engineering, Nazarbayev University, Astana, Kazakhstan, and Department of Computing Science, University of Alberta, Edmonton, Canada.} 
        \and Mohammad Sadegh Talebi\footnote{Department of Computer Science, University of Copenhagen, Copenhagen, Denmark.}
}

\date{\today}

\begin{document}

\maketitle

\begin{abstract}
Power management in multi-server data centers~especially at scale is a vital issue of increasing importance in cloud computing paradigm. Existing studies mostly consider thresholds on the number of idle servers to switch the servers on or off~and suffer from scalability issues. As a natural approach in view~of~the Markovian assumption, we present a multi-level continuous-time Markov decision process (CTMDP) model based on state aggregation of multi-server data centers with setup times that interestingly overcomes the inherent intractability of traditional MDP approaches due to their colossal state-action space.  The beauty of the presented model is that, while it keeps loyalty to the Markovian behavior, it approximates the calculation of the transition probabilities in~a way that keeps~the accuracy of the results at a desirable level. Moreover, near-optimal performance is attained at the expense of the increased state-space dimensionality by tuning the number of levels in the multi-level approach. The simulation
results were promising and confirm that in many scenarios of interest, the proposed approach attains noticeable improvements, namely a near 50\% reduction in the size of CTMDP while yielding better rewards as compared to existing fixed threshold-based policies and aggregation methods. 
\end{abstract}


\section{Introduction}
The inconceivable global surge in computing power demand from video streaming, cryptocurrencies, power-hungry artificial intelligence applications, and numerous cloud-connected devices has changed the operational landscape of data centers. Serving as indispensable powerhouses of the modern digital~era, a large portion of the expenditure is dedicated to cooling data center servers and equipment \cite{Sree2017}. Projected statistics on the~enormous power consumption of data centers reveal a harsh reality in spite of committing to more efficient technologies \cite{Miyuru2016}. In general, a server is said to be \textit{on} when~busy serving jobs. In absence of job requests, a server either remains \textit{idle} or is turned \textit{off}. While idle servers are ubiquitous~in~data centers, each amount for about $50$ to $60$ percent of the energy of its fully utilized state \cite{Gu2020}. Energy-aware cloud data centers minimize the power waste of idle servers~by~either switching them to a low-power \emph{standby} state or the inactive~\textit{off} state. In practice, however, turning the server back to the active~state (physical or virtual machine (VM)) incurs extra power consumption and transition delay, known as \textit{setup} or \textit{spin-up} time, which hinders immediate service to incoming job requests. Over-provisioning servers with jobs adds to the energy costs while under-provisioning may result in delayed service delivery time, thus violating the service-level agreement (SLA). In order to shorten the service delay while saving energy, it is therefore necessary to determine the optimal number of idle and setup servers in cloud data centers under different loads.

\subsection{State-of-the-Art and Prior Work}
The attention drawn towards the theoretical assessment of~energy management in multi-server systems has grown significantly over recent years. 
In \cite{hogade2021energy}, a game-theoretic approach~is~proposed for the workload management in geographically~distributed data centers considering the data transfer costs and queueing delay. 
In \cite{li2022dynamic}, a dynamic VM consolidation algorithm (EQ-DVMCA) is proposed which strikes the balance between the energy consumption and quality of service (QoS)~in cloud data~centers and provides efficient consolidation of virtual resources.~
In \cite{Thomas2019}, a hysteresis queuing model is presented to minimize power costs in cloud systems without explicitly considering server setup delay. 
The mean power consumption in systems with exponential setup time is studied in \cite{Gandhi2010} wherein various operational policies such as the \textsc{On/Idle} policy that turns no server off, the \textsc{On/Off} policy that turns off all servers immediately after becoming idle and has no limit on the number of setup servers, the \textsc{On/Off/Stag} policy that allows at most one server to be in setup at any point of time, and finally, the \textsc{On/Off/$k$ Stag} policy that permits at most $k$ servers to be in setup are investigated. In \cite{phung2019delay}, the authors analyzed the $k$-staggered policy that permits some servers to remain idle after setup using a three-dimensional continuous-time Markov chain (CTMC). The power consumption and waiting time distributions for the \textsc{On/Off} policy have been studied in \cite{Gandhi2013}. The work in \cite{Tuan2017} focused on the system queue length distribution and considered added policies that turn off servers with a finite delay, and also permit servers to go into the sleep mode which, compared to the off mode, induces lower setup time and power usage. In \cite{Cheng2021}, the authors defined two priority queues for different levels of delay-sensitivity and in the case of peak loads, the jobs with lower priority were deferred to promote QoS. However, all these policies  employ static or fixed-threshold approaches to manage the energy consumption of servers.


By partitioning the homogeneous physical machines into~three pools (\emph{hot}, \emph{warm}, and \emph{cold}) with different power levels, the~scalable model introduced in \cite{Longo2011} used interacting Markov chains and fixed-point iteration to derive the mean waiting time and power consumption. In \cite{Wang2015} and \cite{Chang2016}, the authors introduced~heterogeneity in the CPU cores requested by each VM, where the numbers of cores conform to uniform and general distributions, respectively. General spin-up time distribution was investigated in \cite{Gebrehiwot2016} and the mean performance measures and energy consumption of the switching policies were modeled using the \textit{M}/\textit{G}/1 queuing discipline. Though insightful, none of the above efforts explicitly prioritize delay over power consumption. 

The Markov decision process (MDP) framework \cite{puterman2014markov} offers a powerful mathematical approach to deriving optimal power-switching strategies in multi-server systems; see, e.g., \cite{Yang2011,Esa2018,phung2020delay,AALTO2019102034,Maccio2018}.
In \cite{Yang2011}, a MDP model is used to  minimize energy costs and rejected jobs in an Infrastructure-as-a-Service (IaaS) cloud system. A near-optimal solution is~proposed in \cite{Esa2018} for power switching of two dynamic servers to minimize power consumption, delay, and wear-and-tear costs using MDP and look-ahead approach. The MDP approach is also used in \cite{phung2020delay} to find the optimal policy for cost-performance trade-off in virtualized data centers, where VMs are modeled by process sharing queues. Moreover, the size of the proposed MDP is reduced using fixed thresholds to categorize the load of servers into three light, moderate, and high levels. In \cite{AALTO2019102034}, an optimal job routing policy is proposed based on the Whittle index in a system of parallel servers, where each server follows an \textsc{On/Off} policy and is equipped with an infinite buffer. The authors in \cite{Maccio2018} conducted a detailed study on the performance of different multi-server power management policies. Denoting the job arrival and service rates as $\lambda$ and $\mu$, respectively, they showed that keeping $\lambda/\mu + \sqrt{\lambda/\mu}$ servers always on results in a near-optimal solution.

Despite their potential and competence in modeling resource management in multi-server systems, the MDP approaches suffer from curse of dimensionality in the case of hyperscale cloud data centers. More precisely, the associated state-action spaces could grow very large even for a moderate number of servers and queue sizes, thus making standard solution methods (e.g., value iteration \cite{puterman2014markov}) for solving MDPs intractable.~To overcome the curse of dimensionality, \emph{state aggregation} and \emph{state abstraction} methods have been widely studied in the literature (e.g., \cite{Ren2002, li2006towards, Hutter2016, abel2016near, Taylor2008, Abel2016, Saldi2015}), which aim to derive smaller MDPs via merging states that are similar in terms of, e.g., model parameters (transition probabilities and rewards), value functions, etc. So far, a plethora of state aggregation and abstraction methods have been presented and analyzed in the literature in both  discounted (e.g., \cite{Taylor2008, Abel2016}) and undiscounted (e.g., \cite{Saldi2015}) settings. However, these methods often consider generic MDPs and are thus oblivious to the structure of MDPs arising in data canter problems. To our best knowledge, there exists no reported work on state aggregation for optimal energy management in large-scale data centers that incorporates the intrinsic structure of the underlying system. 

Finally, it is worth mentioning that some studies investigate resource management in data centers using reinforcement learning (RL) approaches \cite{ran2019deepee,yi2019toward,zhang2017electricity,islam2015online,ran2022optimizing}. These approaches were used to deal with uncertainty in system parameters or to combat 
the curse of dimensionality due to large state-action spaces. For instance, the DeepEE optimization framework in \cite{ran2022optimizing} utilizes deep RL for jointly optimizing energy consumption in task scheduling and cooling control within a data center. In \cite{yi2019toward}, a deep RL-based  allocation algorithm in the presence of long-lasting and compute-intensive jobs is derived. Methods to minimize energy costs associated with performing tasks with deadlines on data center servers have been investigated in \cite{zhang2017electricity}, which involves scheduling the tasks during periods of low energy costs. Lastly, the authors of \cite{islam2015online} introduce an online resource management framework, termed energy budgeting, specifically designed for virtualized data centers. 
While use of RL approaches allows one to deal with unknown system parameters, it should be stressed that derived resource management policies in these works often fail to admit 
performance guarantees, and could be far from optimal in terms of rewards. 

\subsection{Main Contributions}
In this paper, a multi-server system is considered where the servers experience setup delays and the jobs arrive according to a Poisson process. In the proposed  model, a power manager is responsible for turning the servers on or off to manage the power usage of the entire system or the number of waiting jobs. We aim to find an approximately optimal power-switching policy in such a system that not only minimizes the weighted sum of power consumption of the servers and average delay of the jobs, but also characterizes the trade-off between the dimensionality of~the system state space and closeness to the optimal performance.
To this end, we present the following contributions: 
\begin{itemize}
    \item We introduce a \textit{continuous-time MDP} (CTMDP) formulation of~the system, called the \emph{basic CTMDP}, which~facilitates the derivation of an optimal power-switching policy in our system model, and can be of independent interest~as an accurate model for resource allocation in cloud data centers.
    \item The basic CTMDP suffers from the curse of dimensionality. By employing state aggregation, we propose an efficient approximate CTMDP, referred to as \textit{multi-level CTMDP}, to reduce the size of state-action space of the basic CTMDP. However, unlike the classical state aggregation approaches, we leverage the intrinsic structure in the basic CTMDP to perform aggregation more efficiently. Though increasing the number of levels in the multi-level CTMDP yields higher dimensionality, the performance of the resulting policy becomes closer to the optimal policy derived from the basic CTMDP model.
    \item The proposed multi-level CTMDP is benchmarked against the uniform state aggregation method and fixed-threshold policies in \cite{Maccio2018} under different settings with precedence of delay over power. We show the better performance of our model in terms of the achieved expected average rewards. Our simulation results confirm that in many~scenarios of interest, the proposed approach attains a near $50\%$ reduction in the size of CTMDP while yielding better rewards as compared to conventional methods.
\end{itemize} 


The rest of this paper is structured as follows. Section~\ref{sec2} presents the system model and assumptions. Section~\ref{sec3} details the basic CTMDP as the optimal solution, followed~by the~proposed state-aggregated multi-level CTMDP in Section~\ref{sec4}. The transition rate function and reward function of the multi-level CTMDP are derived in Section~\ref{sec5}. Numerical results are discussed in Section~\ref{sec6}. Finally, Section~\ref{sec7} concludes the paper.

\section{System Model}
\label{sec2}
Consider a data center comprising $C$ servers that serve jobs arriving at the system and a power manager that switches the servers on and off independently. When on, a server is in one~of the three states: \textsc{Busy}, \textsc{Idle}, or \textsc{Setup}. Likewise, a server~is~in the \textsc{Off} state if powered off by the manager. 
The Poisson process has been shown to be an acceptable approximation for job arrivals in data centers \cite{di2014characterizing}. Thus, we assume job arrivals follow a Poisson process with rate $\lambda$. 
Furthermore, the service time of each job is exponentially distributed with rate $\mu$. A server can process a job immediately only if it is in the \textsc{Idle} state. Jobs that fail to receive service instantly upon arrival wait in a finite queue of capacity $Q$ until served in a first-in-first-out (FIFO) manner. The state transition diagram of a server is shown in \figurename{~\ref{fig1}}.
\begin{figure}[!t]
	\centering
	\includegraphics[width=.7\columnwidth]{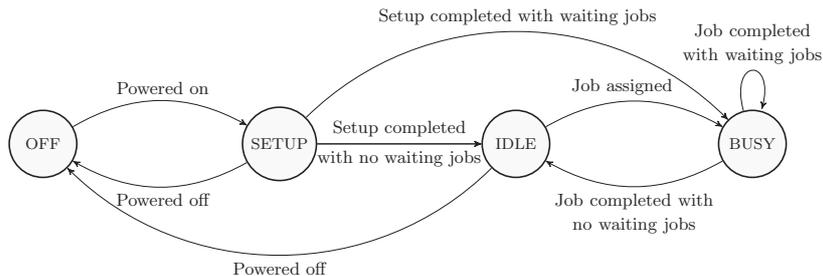}
	\vspace{-1.3em}
	\caption{Power state transition diagram of a data center server.}
	\label{fig1}
\end{figure}

When powered on, the server enters the \textsc{Setup} state and stays there for some amount of time called the \textit{spin-up} time,~which follows an exponential distribution with rate $\gamma \!>\! 0$. Upon~completion of the setup process or the service process of a \textsc{Setup} or \textsc{Busy} server, respectively, the server transitions to the \textsc{Busy} state if there is a head-of-line (HoL) job in the queue or enters the \textsc{Idle} state, if otherwise. Finally, an \textsc{Idle} server either becomes \textsc{Busy} if a job is assigned to it or is switched off by the power manager. Note that the power manager can only turn off servers that are either in \textsc{Idle} or \textsc{Setup} states. Once entering the queue, a newly arrived job is served immediately only if  no other job exists in the queue and an \textsc{Idle} server is available. In other words, the manager assigns the HoL job to an \textsc{Idle} server as soon as it is made available. This happens whenever a \textsc{Setup} server finishes its setup process or a \textsc{Busy} server finishes its service time. 

Although having more \textsc{Idle} and \textsc{Setup} servers results in higher power consumption, it also leads to fewer waiting jobs and more immediate services. In this regard, the average number of \textsc{Idle} and \textsc{Setup} servers can serve as an indicator of the \textit{power penalty} of the system, while the average number of~waiting jobs can be interpreted as the \textit{performance penalty}. Our aim is to determine a power-switching policy, governed by the power~manager, that minimizes a weighted sum of the power and performance penalties.

\section{Basic CTMDP Formulation}
\label{sec3}
The model presented in Section \ref{sec2} constitutes a dynamical~system, where its state evolves as a Markov process, due to the memoryless property of the assumed arrival, setup, and service processes. Moreover, the decisions are made at continuous~time instants when the system state changes. Hence, we may characterize the decision process as a CTMDP \cite{guo2009continuous,puterman2014markov}. An infinite-horizon CTMDP $M$ under the average-reward criterion is a tuple $M=(\mathcal S, \mathcal A, q,r)$, where $\mathcal{S}$ denotes the (finite) state space, and where $\mathcal A = (\mathcal A_s)_{s\in \mathcal S}$ denotes the (finite) action space with $\mathcal A_s$ defining the set of actions available at state $s$. Further, $q$ denotes the transition rate function such that $q(s'|s, a)$ is the transition rate to state $s'\in \mathcal S$ when executing action $a\in \cA_s$ in state~$s\in \mathcal S$. Finally, $r$ denotes the reward function such that $r(s, a)$ is the reward obtained when selecting action  $a\in \cA_s$ in state $s\in \mathcal S$. The various components of the CTMDP $M$ corresponding to the model in Section \ref{sec2} will be specified below. 

{\textbf{The State Space $\cS$.}} We define the system state as $s \!\triangleq\! (b,i)$, where $b \!\in\! \{0, \ldots, C\}$ denotes the number of \textsc{Busy} servers, whereas $i \!\in\! \{-Q,-Q+1,\ldots,C\}$ represents the number of \textsc{Idle} servers if $i \!\geq\! 0$, or the \emph{negated} number of waiting jobs if $i\!<0$. Hence, $|i|$ indicates the number of \textsc{Idle} servers or waiting jobs, and its sign determines whether we have \textsc{Idle} server(s) or waiting job(s). Hereafter, we call $b$ the \emph{B-component} of $s$ with `B' signifying `\textsc{Busy}', whereas we refer to $i$ as the \emph{I/W-component}, where `I' and `W' stands for `\textsc{Idle}' and `waiting', respectively.

{\textbf{The Action Space $\cA$.}} 
At any state, the manager decides on the number of servers put in the \textsc{Setup} and \textsc{Idle} modes at the instants of job arrival, job completion, and setup completion. Concretely, at $s\!=\!(b,i)\in \cS$, taking action $a\in \cA_s$ corresponds to having $a$ \textsc{Setup} servers in $s$ if $a\ge 0$, or to turning off $|a|$ \textsc{Idle} servers if $a<0$. In other words, when $a\ge 0$, the manager sets exactly $a$ servers in the \textsc{Setup} mode. To this end, if the~number of current \textsc{Setup} servers at the decision time exceeds $a$, then the extra servers are powered off. Otherwise, some servers are turned on to have $a$ \textsc{Setup} servers. We distinguish between two cases depending on the sign of $i$. When $i\ge 0$, there are $i$ \textsc{Idle} servers, and hence at most $i$ servers can be turned off so that $a\!\geq\!\!-i$. Moreover, a maximum of  $C\!-\!b\!-\!i$ servers could be in the \textsc{Setup} mode. Thus, $\mathcal{A}_s\!=\!\{-i, \ldots, C\!-\!b\!-\!i\}$. 
On the other hand, when $i\!<\!0$, there is no \textsc{Idle} server (thus, $a\!\geq\!0$) and the manager can only manage the number of \textsc{Setup} servers, which can be at most $C-b$. Hence,  $\mathcal{A}_s\!=\!\{0, \ldots, \!C\!-\!b\}$. 
In general, $\mathcal{A}_s \!=\! \{-i^+,\ldots, C \!-\!b \!-\! i^+\}$, where $i^+ \!\triangleq\! \max{(i,0)}$. As shown in \cite[Theorem 3]{Maccio2018}, the optimal policy always turns on servers following a bulk setup policy; when it decides to turn on some servers, i.e., to put them in the \textsc{Setup} mode, it turns on all \textsc{Off} servers. Hence, we have $\mathcal{A}_s \!=\! \{-i^+,\ldots,0\} \cup \{C \!-\! b \!-\! i^+\}$.


{\textbf{The Transition Function $q$.}} Consider $s=(b,i)\in \cS$, $s'=(b',i')\in \cS$, and $a\in \cA_s$. To determine $q(s'|s, a)$, we consider two cases based on the sign of $a$. For $a\!<\!0$, an \textsc{Idle} server must exist, implying $i\!>\!0$. Since $|a|$ represents the number of \textsc{Idle} servers to be turned off, we have $i + a \geq 0$. Hence,
\begin{equation}
	q(s'|s,a)  \!=\!
	\begin{cases}
		{\lambda};   &  \!b'\!= b+1,~i'\!=\!i+a-1,~i+a>0, \vspace{-0.2em} \\
		{\lambda};   &  \!b'\!= b,~i'\!=\!-1,~i+a=0, \vspace{-0.2em} \\
		{b\mu};      &  \!b'\!= b-1,~i'\!=\!i+a+1,~b > 0.
	\end{cases}
	\label{eq11}
\end{equation}
To verify \eqref{eq11}, note that upon arrival of a new job, the number of remaining \textsc{Idle} servers decreases by one after taking action $a$ (if $i+a\!>\!0$) and hence, the number of \textsc{Busy} servers increases by one. If $i+a\!=\!0$, then $i'\!=\!-1$ implying that the arriving~job waits in the queue. Finally, completion of a job with rate $b\mu$ increases and decreases the number of \textsc{Idle} and \textsc{Busy} servers by one, respectively. When $a\!\geq\!0$, where a total of $a$ servers will be in the \textsc{Setup} mode, we have: 
\begin{equation}
q(s'|s,a)  \!\!=\!\begin{cases}
		{\lambda};   &  b'\!=\!b+1,~ i'\!=\!i-1, ~i\!>\! 0, \vspace{-0.15em} \\
		{\lambda};   &  b'\!=\!b,~ i' \!=\!i-1,~ -Q < i\! \le\! 0, \vspace{-0.15em} \\
		{\lambda};   &  b'\!=\!b,~ i' \!=\!i,~ i\!=\!-Q, \vspace{-0.15em} \\
		{b\mu};      &  b'\!=\!b-1,~ i'\!=\!i+1,~ i\!\ge\! 0, \vspace{-0.15em} \\
		{b\mu};      &  b'\!=\!b,~ i'\!=\!i+1,~ i\!<\! 0, \vspace{-0.15em} \\
		{a\gamma}; 	 &  b'\!=\!b,~ i'\!=\!i+1,~ i\!\ge\! 0, \vspace{-0.15em} \\
		{a\gamma};   &  b'\!=\!b+1,~ i'\!=\!i+1,~ i\!<\! 0.
	\end{cases}
	\label{eq1}
\end{equation}
The first case in \eqref{eq1} is verified by noting that when $i\!\ge\!0$, 
then upon arrival of a new job, the number of \textsc{Idle} and \textsc{Busy} servers decreases and increases by one, respectively. However, in the second case where $i\! \leq \!0$, the number of \textsc{Busy} servers remains unchanged, but the number of waiting jobs increases by one. Specifically, we have $i' = i - 1$ under the condition $-Q <~i$ (i.e., the queue is not full). Otherwise, if $i = -Q$, then the~number of waiting jobs remains unchanged and the newly arriving job is dropped. In the fourth and fifth cases of \eqref{eq1}, a \textsc{Busy} server transitions to an \textsc{Idle} state at rate $b\mu$, while the last two cases signify that a \textsc{Setup} server becomes \textsc{Idle} at a rate of $a\gamma$. In~either case, the new \textsc{Idle} server remains \textsc{Idle} if $i \geq 0$, and transitions back to \textsc{Busy} to serve a waiting job if otherwise. Note that the packet arrival at state $b \!=\! C$ is included in the second and third cases. As such, when $b \!=\! C$, we have $i \leq 0$ and the arrived packet awaits in the queue if it is not full.

{\textbf{The Reward Function $r$.}} Let $\Psi(s,a)$ denote sum of output transition rates at state $s=(b,i)$ under action $a\in \cA_s$: 
$	\Psi(s,a) = a^+\gamma + b \mu + \lambda$. The quantity $1/\Psi(s,a)$ indicates the average time spent at state $s$ under action $a$. Observe that $|i^-|$ (with $x^-\!\triangleq\!\min{(x,0)}$ for any $x$) is linked to the performance penalty, whereas the power penalty may be defined using a weighted sum of the number of \textsc{Idle} servers, $i^+$, and the number of \textsc{Setup} servers, $a^+$. 
In order to define a reward function in line with the objective discussed in Section \ref{sec2}, we define:
\begin{equation}
	r(s,a) = \frac{-1}{\Psi(s,a)}\big(c_{\text{perf}}\,|i^-| + c_{\text{power}} i^+\!+  c_{\text{power}}' a^+\! \big) ,
	\label{eq4}
\end{equation}
where $c_{\text{perf}},c_{\text{power}}$, and $c_{\text{power}}'$ are application-specific positive numbers.
In fact, $r(s,a)$ in (\ref{eq4}) is a weighted sum of performance and power penalties scaled by the average dwell time at $(s,a)$. Evidently, the ratio $c_{\text{power}}/c_{\text{power}}'$ controls the relative emphasis put on the power consumption of \textsc{Idle} and \textsc{Setup} modes.
It is worth remarking that, firstly, we do not model the cost of job dropouts in the reward function since, due to the precedence of performance penalty over power penalty, the optimal policy never results in considerable job dropout probabilities, especially in the presence of a finite but sufficiently large queue size. Secondly, the analytical approach considered here is suited for moderate and light traffics, where the jobs do not occupy all servers almost surely and thus, the implementation of power management policies is justified.

Undertaking the approach in Chapter 11 of \cite{puterman2014markov}, the CTMDP $M$ can be transformed into a discrete-time average-reward MDP and can be solved using standard methods such as policy iteration, value iteration, or linear programming. The following lemma presents the time complexity of solving the associated MDP to $M$.

\begin{lemma}
\label{lem:CTMDP_time}
The (per-step) time complexity of solving the MDP associated to $M$ using value iteration is $O(C^2(Q+C^2)(Q+C))$.
\end{lemma}

\begin{proof}
The time complexity of an iteration in the value iteration algorithm is $O(|\mathcal S|\sum_{s\in\mathcal{S}}|{\mathcal{A}_s}|)$ \cite{puterman2014markov}. Introduce $\mathcal S_+:=\{(b,i)\in \mathcal S, i>0\}$ and $\mathcal S_-:=\{(b,i)\in \mathcal S, i\le 0\}$. Then, $\mathcal S_+$ and $\mathcal S_-$ define a partition of $\mathcal S$ so that  $|\mathcal{S}|= |\mathcal{S}_+|+|\mathcal{S}_-|$. Furthermore,   $\sum_{s\in\mathcal{S}}|{\mathcal{A}_s}|=\sum_{s\in\mathcal{S}_+}|{\mathcal{A}_s}|+\sum_{s\in\mathcal{S}_-}|{\mathcal{A}_s}|$. We consider two cases:
\begin{itemize}
    \item \textit{When $s=(b,i) \in \mathcal S_-$:} Recalling that $i\in\{-Q,\ldots,0\}$ and $b\in \{0,\ldots,C\}$, we have $|\mathcal{S}_-|=(Q+1)(C+1)$. Furthermore, $\mathcal A_s = \{0, a=C\!-\!b\!-\!i\}$ so that  $\sum_{s\in\mathcal{S}_-}|{\mathcal{A}_s}|=2(Q+1)(C+1)$.
    \item \textit{When $s=(b,i) \in \mathcal S_+$:} Recalling that $0 \leq b \leq C-i$, we have $|\mathcal{S}_+|=\sum_{i=1}^{C}(C-i+1)=C(C+1)/2$. Also, ${\mathcal{A}_s}=\{-i,\ldots,0\}\cup\{C-b\}$, and thus, $|{\mathcal{A}_s}|=i+2$.  Consequently, $\sum_{s\in\mathcal{S}_+}|{\mathcal{A}_s}|=\sum_{i=1}^{C}(C+1-i)(i+2)=C(C+1)(C/6+4/3)$.
\end{itemize}
  Putting them together, we obtain  $|\mathcal{S}|\sum_{s\in\mathcal{S}}|{\mathcal{A}_s}|=(C+1)^2(Q+C/2+1)(C^2/6+4C/3+2Q+2) = O(C^2(Q+C^2)(Q+C))$, thus concluding the proof. 
\end{proof}

\section{Approximation via State Aggregation: Multi-level CTMDP}
\label{sec4}
We focus on large-scale data centers, where $C$ represents a significant number. For instance, in December 2014, Amazon Web Services operated approximately $C=1.4$ million servers across 28 availability zones. According to Lemma \ref{lem:CTMDP_time}, computing an optimal policy in the exact model $M$ for such data centers incurs a cost of $O(C^5)$, even for moderate values of $Q$, which~is unfeasibly large.  We remedy this issue by introducing a more manageable approximation of $M$ with a reduced state-action space, known as the multi-level CTMDP. 


\subsection{Assumptions and Approximation}
\label{sec4.1}
We first need to introduce the queuing approximation for our system captured by the following assumptions:

\begin{assumption}
\label{assump:large_C}
The number of servers, $C$, is large~enough to serve all arriving jobs.
\end{assumption}

Assumption \ref{assump:large_C} is justified since the power-switching policy is used for systems with low to moderate traffic. In other~words, each job finds at least one \textsc{Off}, \textsc{Idle}, or \textsc{Setup} server upon its arrival, and the probability of all servers being \textsc{Busy} is~infinitesimal. Otherwise, the power manager would need to keep the servers permanently powered on to accommodate as many jobs as possible, in which case implementing the power-switching policy would be meaningless. 

\begin{assumption}
\label{assump:opt_policy}
The optimal policy always turns on at least~$|i|$ servers, i.e., $a \geq |i|$ for states with waiting jobs ($i < 0$).
\end{assumption}

This assumption bears resemblance to the $k$-staggered policy discussed in \cite{Maccio2018}. Thus, for any state $s=(b,i)$ with $i\!<\!0$, the action space $\mathcal{A}_s\!=\!\{a| -i\!<\!a\!<\!C\!-b\}$, ensuring that for any $s=(b,i)\in\cS$, $\mathcal{A}_s=\{a| -i\!\leq\!a\!\leq\!C\!-b\!-i^+\}$.~For~now, let us assume that the setup delay is zero. In such a case, given that turning the servers on and off does not induce any costs, the optimal power-switching policy turns the servers off upon being \textsc{Idle} since they can become immediately \textsc{Idle} again when needed. Then, considering Assumptions~\ref{assump:large_C} and \ref{assump:opt_policy}, the proposed queueing system can be assumed to have infinite number of servers, and due to exponential inter-arrival and service times, it can be modeled as an $M/M/\infty$ queue. 

In an $M/M/\infty$ queue, the number of \textsc{Busy} servers in the steady state follows the Poisson distribution with rate $\rho \!=\!\lambda/\mu$. On the other hand, in realistic scenarios, with non-zero setup delays, some jobs which do not find an \textsc{Idle} server upon their arrival incur a delay before joining a server, which depends on the system state and the taken action. However, according to Assumption~\ref{assump:opt_policy}, the waiting time never exceeds a setup delay~(in a stochastic sense). In this regard, in the real system, we stick~to the approximation that the number of \textsc{Busy} servers in the steady state, denoted by $b$, follows the Poisson distribution with rate $\rho \!=\!\lambda/\mu$. This assumption helps us to analytically derive the transition probabilities of the proposed multi-level CTMDP, for which our numerical results show that the resulting optimal policy outperforms previously studied counterparts. Besides, such an approximation becomes more accurate when the implemented policy results in smaller waiting time for the jobs by having more \textsc{Setup} and \textsc{Idle} servers. Finally, since $\rho$ is large enough in our problem (e.g., $\rho>10$), the Poisson~distribution behaves similarly to a normal distribution with mean $\rho$ and standard deviation $\sqrt{\rho}$. The derivation of the multi-level CTMDP in the following subsections relies on the symmetry property of the normal distribution.

\subsection{The Multi-level CTMDP: Construction}
\label{sec4.2}
We now formally introduce the multi-level CTMDP. Given the CTMDP $M=(\cS,\cA,q,r)$ introduced earlier, we denote its corresponding multi-level CTMDP by $\Magg=(\Sagg, \Aagg, \qagg, \ragg)$, where `$\mathrm{ml}$' signifies \emph{`multi-level'}. The multi-level CTMDP $\Magg$ is created by aggregating states from $M$, where each state in $\Magg$ represents a combination of multiple states from $M$. A detailed construction of state and action in $\Magg$ follows. 

\textbf{The State Space $\Sagg$.} 
We consider aggregation by partitioning the set of feasible B- and I/W-components via suitably defined intervals. Therefore, it is natural to define the state~in $\Magg$ as $S=(B,I)$, where $B$ and $I$ are related to  \emph{the aggregated} B-components and I/W-components, respectively. 
More specifically, $B$ (resp.~$I$) refers to the \emph{index} of the interval to~which the B-component (resp.~I/W-component) of its aggregated states belongs. We consider $L$ levels for each of the B and I/W~components, for some suitably chosen natural number $L$. 
To define $B$, we partition the set of available servers $\{0,\ldots,C\}$ into $L$ disjoint subsets as follows:
\begin{align*}
\label{eq:interval_B_comp}
\{0,1,\ldots,U_1-1\}, \{U_1,\ldots,U_2-1\},\ldots,\{U_{L-1},\ldots,C\}, 
\end{align*}
where $U_1<U_2<\ldots<U_{L-1}<C$ are real numbers~that will be determined momentarily. The B-component of any state $s\!=\!(b,i)\!\in\!\cS$ belongs to one of the aforementioned sets. 
Let us index the sets with $B \in \{0,1,\ldots,L - 1\}$. Then, $B$ corresponds to the states whose B-components belong to $\{U_{B},\ldots,U_{B+1} -1\}$. In other words, $B$ corresponds to states in $\big\{s=(b,i): b\in \{U_{B},\ldots,U_{B+1}\!-1\}\big\}$. For simplicity, we consider sets of equal size. 
That is to say, $K_\mathrm B\triangleq U_{B+1}\!- U_B$ does not depend on $B$. Hence, to determine $U_1,\ldots,U_{L-1}$ (and thus, $B$), it suffices to determine $K_\mathrm B$. We shall refer to $K_\mathrm B$ as \emph{\textsc{Busy} level size}. We use the $(1 - \epsilon)$-confidence interval\footnote{$(1\!-\!\epsilon)$-confidence interval is the interval to which the B-component belongs with probability greater than $1-\epsilon$.} for some small enough $\epsilon$ (e.g., $\epsilon\!=\!0.01$) to set  $K_\mathrm B$. More precisely, recalling that the B-component in our model denotes the number of \textsc{Busy} servers, it follows a Normal distribution as explained in Section \ref{sec4.1}. Hence, the above confidence interval is derived as $\bigl[F^{-1}\!\left(\frac{\epsilon}{2};\rho\right),F^{-1}\!\big(1 \!-\!\frac{\epsilon}{2};\rho\big)\bigr]$, where $F^{-1}(\cdot;\rho)$ is the inverse of $F(\cdot;\rho)$, and $F(\cdot;\rho)$ denotes the CDF of the Poisson distribution with rate $\rho=\lambda/\mu$. 
Now, the \textsc{Busy} level size $K_\mathrm B$ is defined as:  
\begin{equation}
	K_{\mathrm B}=-\left \lceil\Big(F^{-1}\!\left(\tfrac{\epsilon}{2};\rho\right)-F^{-1}\!\big(1 \!-\!\tfrac{\epsilon}{2};\rho\big)\Big)/L\right \rceil \,,
	\label{eq5}
\end{equation}
where $\lceil\cdot\rceil$ is the ceiling function ensuring an integer level size. Further, we define  
$\underline \beta = \left\lfloor \rho-\tfrac{1}{2}K_\mathrm BL\right\rfloor^+$ and $\overline \beta = \left\lceil\rho+\tfrac{1}{2}K_\mathrm BL\right\rceil$, 
where $\lfloor\cdot\rfloor$ denotes the floor function. Hence, the endpoints $U_1,\ldots,U_{L-1}$ are obtained as\footnote{Here, we use the symmetry of the Normal distribution of the B-component around $\rho$ to define the new lower and upper endpoints of the confidence interval.}
$
U_B = B K_\mathrm B+\underline \beta, 
$
for $B=1,\ldots,L-1$. The construction above implies that all values of $b<\underline{\beta}$ are allocated to the level $B=0$, whereas those with $b\!>\!\overline \beta$ are allocated to the level $B=L\!-\!1$. 


To define the I/W-component $I$, we undertake a similar approach to partition $\{-Q,\ldots,C\}$. Specifically, we partition the positive part (i.e., $\{0,\ldots,C\}$) into $L$ subsets, each of size $K_{\mathrm I} = C/L$.\footnote{For simplicity, we use the same number of partitions as for B-component.} Then, we partition the negative part (i.e., $\{-Q,\ldots,-1\}$) into levels of size $K_\mathrm I$. Thus, the negative part of I/W-component has $\left\lceil Q/K_{\mathrm I} \right\rceil$ levels so that $I \!\in\! \{\!-\!\left\lceil Q/K_\mathrm I \right\rceil,\ldots, L\!-\!1\}$, where the level $I$ aggregates states with $i \!\in\! \{IK_\mathrm I \!+\! 1,\ldots,(I \!+\! 1)K_\mathrm I\}$.

\textbf{The Action Space $\Aagg$.} We define action space in $\Magg$ by aggregating actions in $M$. Let $\Aagg_S$ denote the set of actions available at $S=(B,I)$ in $\Magg$. Following a similar construction as states to define $\Aagg_S$ using a level size $K_\mathrm I$, we define: 
$$
\Aagg_{S}\!=\!\big\{\!-\!I,\ldots,0\big\}\cup \big\{\left(C\!-\!U_B\!-\!I^+K_\mathrm I\right)/K_\mathrm I\big\} \,,
$$
where $C \!-\! U_B \!-\! I^+K_\mathrm I$ is the number of available \textsc{Off}~servers, assuming $U_B$ servers are \textsc{Busy} and $I^+K_\mathrm I$ servers are \textsc{Idle}. When $A \!>\! 0$, $AK_\mathrm I$ servers are powered on to be in the \textsc{Setup} state; else, all \textsc{Setup} and $AK_\mathrm I$ \textsc{Idle} servers are turned off.

As mentioned earlier, the trade-off between optimality and dimensionality can be adjusted by choosing different values for $L$. In particular, a larger $L$ yields a larger (aggregated) state space while ensuring that the optimal policies in $M$ and $\Magg$ become closer in terms of expected average reward. 

\section{Transition and Reward Functions of $M^{\scriptsize}$}
\label{sec5}
This section is devoted to deriving the transition rate function $\qagg$ and reward function $\ragg$ of the multi-level CTMDP $\Magg$.
In view of the construction of $\Magg$, this task entails calculating the transition rates between various levels. 
In order to make the presentation more tractable, we begin with calculating a few key quantities that prove instrumental in formulating $\qagg$ and $\ragg$. 

\subsection{Preliminaries: Level Boundary Probabilities}
\label{sec:premilinaries}
The transition between levels occurs when the I/W-component or B-component (in $M$) are at the boundaries of their corresponding levels. Therefore, we need to derive the probability of being at the boundaries of \textsc{Busy} and I/W levels. 

We recall that $F(\cdot;\rho)$ denotes the CDF of the Poisson distribution with rate $\rho=\lambda/\mu$, which is in fact the distribution of the number of \textsc{Busy} servers as a result of the assumption of Poisson arrivals; see Section~\ref{sec4.1}. Further, let $f(\cdot;\rho)$ denote the probability mass function (pmf) of $F(\cdot;\rho)$. For brevity, we omit the dependence of $F$ and $f$ on $\rho$ as it is fixed throughout.  

Considering $s\!=\!(b,i)$, we observe that the probability that $b\in\big\{U_{B},\ldots, U_{B+1}-1\}$ (i.e., the number of \textsc{Busy} servers belongs to the $B$-th level) is $F\big(U_{B+1}-1\big) - F\big(U_{B}-1\big)$. Define
\begin{align*}
	\underline p(B)&= \frac{ f(U_B)}{F(U_{B+1}\!-\!1) - F(U_{B}-1)}, \\ 
	\overline p(B)&= \frac{ f(U_{B+1}-1)}{F(U_{B+1}\!-\!1) - F(U_{B}\!-\!1)}.
	\label{eq1112}
\end{align*}
It is evident that $\underline p(B)$ (resp.~$\overline p(B)$) is the probability that $b$ coincides with the lower (resp.~upper) boundary of the \textsc{Busy} level $B$ in $\Magg$. Furthermore, the average number of \textsc{Busy} servers at level $B$, denoted by $N_B$, is:
\begin{equation}
	N_B = \frac{1}{F\big(U_{B+1}-1\big) - F\big(U_{B}-1\big)}\sum_{x=U_{B}}^{U_{B+1}-1}x f(x)\, .
	\label{eq9}
\end{equation}

For the I/W-component, we are interested in computing the probability that $M$ is at the lower (resp.~upper) boundary of I/W level $I$ conditioned on the event that $\Magg$ is in state $S$ and action $A$ is chosen. These probabilities are denoted by $\underline u(S,A)$ and $\overline u(S,A)$, respectively, and derived in the following lemma.
\begin{figure}[!t]
	\centering
	\includegraphics[width=.8\columnwidth]{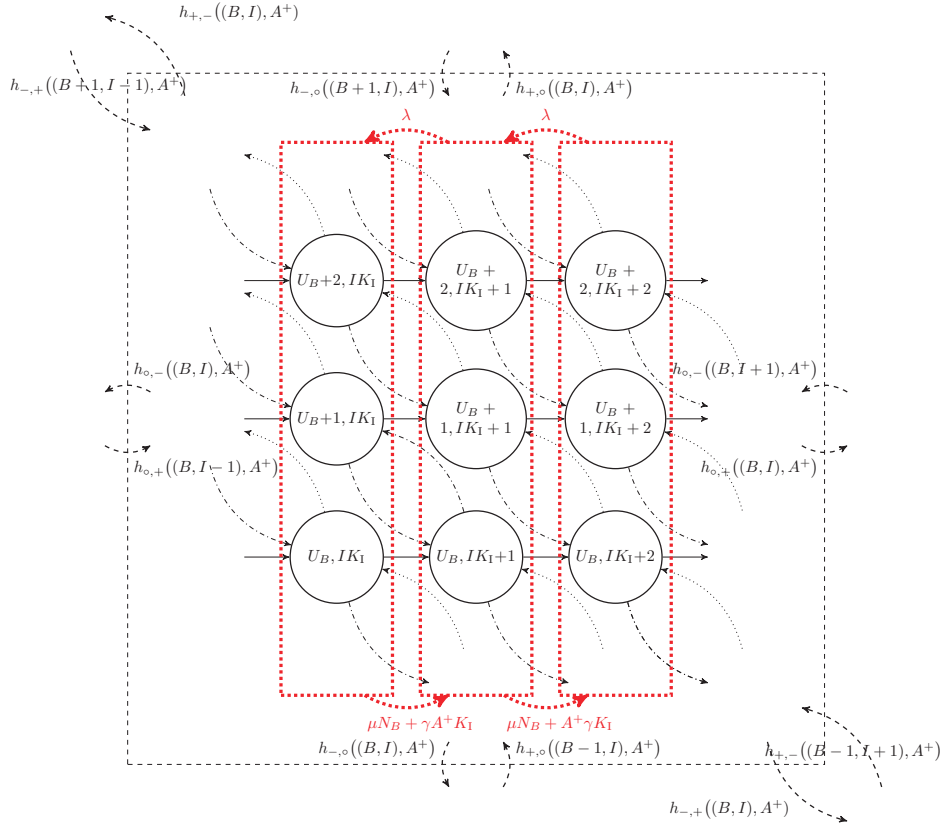}
	\vspace{-1.3em}
	\caption{Transition diagram of a multi-level state $(B,I)$ when $I \!>\! 0$ and $K_\mathrm B \!=\! K_\mathrm I \!=\! 3$. Dotted transitions correspond to a job arrival with rate $\lambda$, solid transitions correspond to a setup completion with rate $\gamma A^+ K_\mathrm I$, and dash-dotted transitions correspond to a service completion with rate $b\mu$, where $b$ is the number of \textsc{Busy} servers in the current state.}
    \label{fig2}
\end{figure}

\begin{lemma}
\label{lem:u_probabilities}
Let $\eta(S,A) = (\mu N_B+A^+K_{\mathrm I}\gamma)/\lambda$. If $\eta(S,A) \neq 1$, 
\begin{align*}
	\underline u(S,A)\!=\!\frac{1-\eta(S,A)}{1-\eta(S,A)^{K_{\mathrm I}}}\,, \,\,
    \overline u(S,A)\!=\!\eta(S,A)^{K_{\mathrm I}-1} \underline u(S,A).
\end{align*}
Furthermore, $\underline u(S,A) = \overline u(S,A) = K_\mathrm I^{-1}$ when $\eta(S,A) = 1$. 
\end{lemma}

\begin{proof}
Let $S=(B,I)\in \Sagg$ and $A\in \Aagg$. The proof relies on constructing an approximate birth-death (BD) process associated to $(S,A)$. Recall that by construction, for each $IK_\mathrm I \!\leq\!i\!\leq\!(I\!+\!1)K_\mathrm I\!-\!1$, we aggregate all states $s=(b,i)\in \mathcal S$ with $U_B \!\leq\! b \!\leq\! U_{B+1}-1$ into one state, called \emph{meta-state $i$}. \figurename~\ref{fig2} and \figurename~\ref{fig3} show the transition probabilities of the states in the multi-level state $S=(B,I)$ when $I \!>\! 0$ and $I \!\leq\! 0$, respectively. In fact, we aggregate all states within the red dotted boxes in these two figures into one state to obtain $K_{\mathrm I}$ meta-states. Also, the arrival (resp.~departure) rate of each meta-state is equal to the average of the arrival (resp.~departure) rates of the corresponding aggregated states. When deriving these average rates, we ignore all transitions from other multi-level states into $S=(B,I)$ and vice versa. These transitions, in fact, belong to the boundary states $(b,i)\!\in\! \{U_B,\ldots,U_{B+1}-1\} \!\times\! \{IK_{\mathrm I},\ldots, (I+1)K_{\mathrm I} \!-\! 1\}$. With this approximation, the average transition rate from each meta-state $i$ to the next state $i \!+\! 1$ becomes $\gamma A^+  K_{\mathrm I}  + N_B$, while that from each meta-state $i$ to the previous state $i \!-\! 1$ is equal to $\lambda$. Since these birth and death rates are the same in all meta-states, the resulting approximated Markov chain forms a BD process with the birth and death rates equal to $\gamma A^+K_{\mathrm I}\!+\!N_B$ and $\lambda$, respectively. Let $P(i|S,A)$ denote the probability that the I/W-component (in $M$) is $i$ given that $\Magg$ is in $S\in \Sagg$ and $A\in \Aagg_S$ is taken.  Then, in the BD process, we have:
\begin{equation*}
	\lambda P(i+1|S,A) = \left(\mu N_B+\gamma A^+ K_{\mathrm I}\right) P(i|S,A) \,.
\end{equation*}
so that $P(i+1|S,A) = \eta P(i|S,A)$, where for brevity we omit the dependence of $\eta$ on $(S,A)$. 

\begin{figure}[!t]
	\centering
	\includegraphics[width=.7\columnwidth]{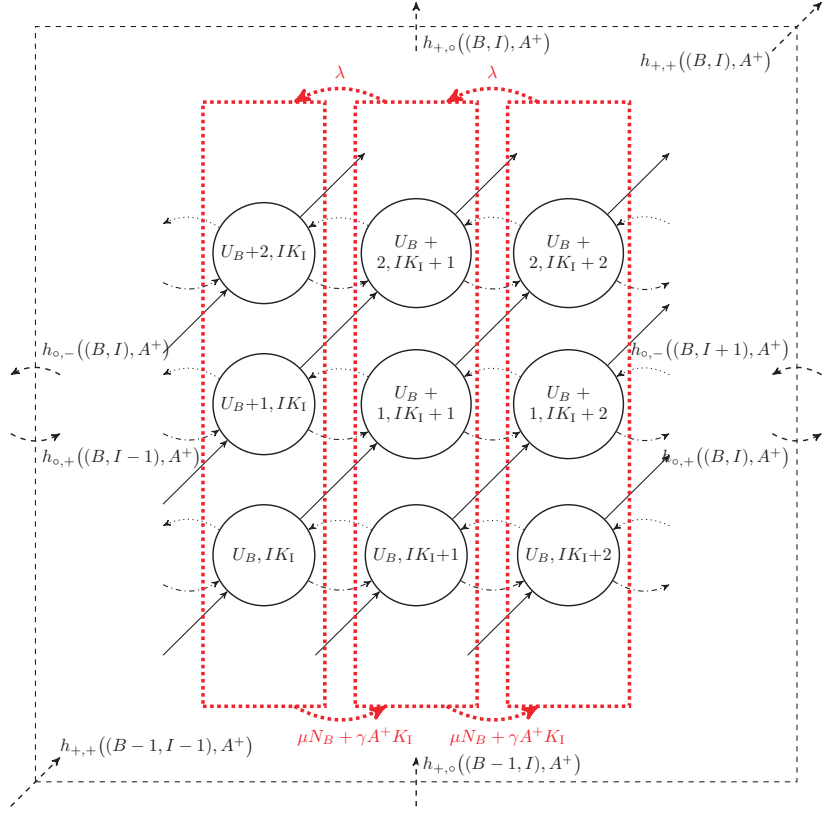}
	\vspace{-.5em}
	\caption{Transition diagram of a multi-level state $(B,I)$ when $I \!\leq\! 0$ and $K_\mathrm B \!=\! K_\mathrm I \!=\! 3$. Dotted transitions correspond to a job arrival with rate $\lambda$, solid transitions correspond to a setup completion with rate $\gamma A^+K_\mathrm I$, and dash-dotted transitions correspond to a service completion with rate $b\mu$, where $b$ is the number of \textsc{Busy} servers in the current state.}
    \label{fig3}
\end{figure}
We claim $\sum_{k=0}^{K_{\mathrm I}-1}\eta^{k}P(IK_{\mathrm I}|S,A) = 1$. This can be verified by observing that the left-hand side is the probability that $\Magg$ is at level $I$ for the given state $S$ and action $A$. Hence, for $S=(B,I)$, this happens with probability $1$ and the claim follows. Using algebraic manipulations, it then implies that for $i=IK_\mathrm I$: 
	$$
	P(IK_{\mathrm I}|S, \!A) =\!\begin{cases}
								\!1/K_{\mathrm I}; & \eta \!=\! 1, \\
								\!\frac{1-\eta}{1-\eta^{K_{\mathrm I}}}; & \text{otherwise.}
							\end{cases}
	$$
	Following a similar reasoning, we obtain  
	\begin{align*}
	P((I+1)K_{\mathrm I}-1|S&, \!A) =\begin{cases}
								\!P(IK_{\mathrm I}|S, \!A)K_{\mathrm I}; & \eta \!=\! 1, \\
								\!P(IK_{\mathrm I}|S, \!A)\frac{1-\eta^{K_{\mathrm I}}}{1-\eta}; & \text{otherwise.}
							\end{cases}
	\end{align*}
	Observing that $\underline u(S,A) = P(IK_{\mathrm I}|S, \!A)$ and $\overline u(S,A) = P((I+1)K_{\mathrm I}-1|S, \!A)$ concludes the proof. 
\end{proof}
 
We conclude this subsection by deriving the expected number of either \textsc{Idle} servers or waiting jobs under a given pair $(S,A)$, which we denote by $\overline I(S,A)$. We have:
\begin{align}
	\overline I(S,A) &= \sum_{i = K_{\mathrm I}I}^{K_{\mathrm I}(I+1)-1}i P(i|S,A) \nonumber \\
	&= \frac{1-\eta(S,A)}{1-\eta(S,A)^{K_{\mathrm I}}} \sum_{i = K_{\mathrm I}I}^{K_{\mathrm I}(I+1)-1}i\eta(S,A)^{i-K_{\mathrm I}I} \nonumber \\
	&=\underline u(S,A) \sum_{i = K_{\mathrm I}I}^{K_{\mathrm I}(I+1)-1}i\eta(S,A)^{i-K_{\mathrm I}I}
	\label{eq:I_agg_average}
\end{align}
where the second line follows from the recursive property of $P$ established in the proof of Lemma \ref{lem:u_probabilities}.

\subsection{The Transition Function $\qagg$}
\label{sec5.2}
We are now ready to fully characterize $\qagg$ using the level boundary probabilities in the preceding subsection. This entails calculating the transition rates when the \textsc{Busy} level, I/W level, or both change. In what follows, we derive the level transition probabilities resulting from the positive part of the action, i.e., $A^+$. These transitions, corresponding to the cases $I \!>\! 0$ and $I \!\leq\! 0$, are shown on the outer dotted boxes in \figurename~\ref{fig2} and \figurename~\ref{fig3}, respectively. The value of $A^-$ does not affect the \textsc{Busy} level $B$ and only changes the I/W level $I$ to $I \!+\! A^-$ instantly with probability one. Hence, its effect is incorporated into the conditions. The transition function $\qagg$ admits the following form
\begin{equation*}
	\qagg(S'|S,\!A) \!=\! \begin{cases}
						\!h_{+,\circ}(S,A^+\!); &  \!B'\!=\!B\!+\!1,I'\!=\!I\!+\!A^-\!, \\
						\!h_{-,\circ}(S,A^+\!); &  \!B'\!=\!B\!-\!1,I'\!=\!I\!+\!A^-\!, \\
						\!h_{\circ,+}(S,A^+\!); &  \!B'\!=\!B,I'\!=\!I\!+\!1\!+\!A^-\!, \\
						\!h_{\circ,-}(S,A^+\!); &  \!B'\!=\!B,I'\!=\!I\!-\!1\!+\!A^-\!, \\
						\!h_{+,+}\!(S,A^+\!); & \!B'\!=\!B\!+\!1,I'\!=\!I\!+\!1\!+\!A^-\!, \\
						\!h_{+,-}(S,A^+\!); &  \!B'\!=\!B\!+\!1,I'\!=\!I\!-\!1\!+\!A^-\!, \\
						\!h_{-,+}(S,A^+\!); &  \!B'\!=\!B\!-\!1,I'\!=\!I\!+\!1\!+\!A^-\!.
\end{cases}
\end{equation*}
where $h_{i,j}$ with $i,j\in\{-,+,\circ\}$ are (output) rate functions that will be derived momentarily. Here, the first subscript $i$ represents a change in $B$, and the second $j$ captures a change in $I$. The symbol $\circ$ indicates no change during the transition, whereas $+$ and $-$ indicate increment and decrement, respectively. For example, $h_{+,\circ}(S,A)$ denotes the transition rate from state $S\!=\!(B,I)$ to state $S'=(B\!+\!1,I)$ under  action $A$. The rest of this subsection is devoted to deriving the rate functions $h_{i,j}$. Here $|A^-|$ indicates the number of \textsc{Idle} servers that will be turned off. Thus, as an immediate result of taking action $A$, the number of \textsc{Idle} servers is reduced instantaneously by $|A^-|$, i.e., $I' \!=\! I \!+\! A^-$. Obviously, if $A \!\geq\! 0$, then $I'$ will not change as a result of $A^-$. Other changes to $I$ occur as a result of action, $A^+$, and in a probabilistic manner. 

\textbf{\textsc{Busy} Level Transition Rates ($h_{+,\circ}$ and $h_{-,\circ}$).} 
\label{sec4.4.1}
If a transition from state $S\!=\!(B,I)$ to state $S'=(B\!+\!1,I)$ under the action $A^+$ occurs, then the number of \textsc{Busy} servers should be equal to the upper boundary of the current \textsc{Busy} level $B$ (i.e., $b\!=\!U_{B+1}-1$) and increases  by one as well.  This increment occurs when a job arrives with rate $\lambda$ when $I \ge 0$ or a \textsc{Setup} server becomes \textsc{Idle} with rate $\gamma A^+K_{\mathrm I}$ when $I < 0$. In both cases, the I/W-component remains unchanged if that of $M$ is different than the lower and upper boundaries of the current I/W-component, respectively. Hence, 
\begin{equation}
	h_{+,\circ}(S,A^+)	\!=\!	\begin{cases}
								\!\lambda \overline p(B)\left(1 \!-\! \underline u(S,A)\right); & I \!\ge\! 0, \\
								\!\gamma A^+K_{\mathrm I}  \overline p(B) \left(1 \!-\! \overline u(S,A)\right); & I \!<\! 0.
							\end{cases}
	\label{eq19}
\end{equation}

The transition rate from $S\!=\!(B,I)$ to $S'=(B\!-\!1,I)$ under $A^+$ is captured by $h_{-,\circ}(S,A^+)$. When $I \!\geq\! 0$, such a transition can arise only if the number of busy servers is equal to the lower boundary of the current \textsc{Busy} level $B$ (with probability $\underline p(B)$), the number of \textsc{Idle} servers is not equal to the upper boundary of the current I/W level $I$ (with probability $1-\overline u(S,A)$), and a \textsc{Busy} server turns \textsc{Idle} (with rate $\mu U_B$). We thus have:
\begin{equation}
	h_{-,\circ}(S,A^+)	\!=\!	\begin{cases}
								\mu U_B \underline p(B)\left(1 \!-\! \overline u(S,A)\right); & I \ge 0, \\
								0\,; & I < 0.
							\end{cases}
	\label{eq20}
\end{equation}

\textbf{I/W Level Transition Rates ($h_{\circ,+}$ and $h_{\circ,-}$).} 
\label{sec4.4.2}
Introduce
\begin{equation}
	N_B^- = \frac{1}{F\big(U_{B+1}-1\big) - F\big(U_{B}-1\big)}\sum_{x=U_B+1}^{U_{B+1}-1}x f(x).
	\label{eq21}	 	
\end{equation}
In fact, $N_B^-$ captures the average number of \textsc{Busy} servers in level $B$ without including the lower boundary of the level. Transitions from $S\!=\!(B,I)$ to $S'=(B,I\!+\!1)$ under $A^+$ occurs at rate:
\begin{equation*}
	h_{\circ,+}(\!S,\!A^+\!)\!=\!	\begin{cases}
							\!\overline u(S,\!A) \!\big(\gamma A^+\!K_{\mathrm I} \!+\! (\!1\!-\!\underline p(B))\mu N_B^-\!\big); & \!\!\!I \!\ge\! 0,\\
							\!\overline u(S,\!A) \!\big(\mu N_B \!+\! (\!1\!-\!\overline p(B)) \gamma A^+\!K_{\mathrm I}\big);  & \!\!\!I \!<\!0.
						\end{cases}
\end{equation*}
To verify this, observe that when $I \!\geq\! 0$, $I$ increases by one if $i$ is equal to the upper boundary of current I/W level (with probability $\overline u(S,A)$) and increases by one as well (see \figurename~\ref{fig2}). Also, $i$ increases if either a \textsc{Setup} server becomes \textsc{Idle} with rate $\gamma A^+K_{\mathrm I}$ or a \textsc{Busy} server turns \textsc{Idle}. In the latter, to remain in the same \textsc{Busy} level, the number of busy levels should not be equal to the lower boundary of the current \textsc{Busy} level. Thus, the rate of having an \textsc{Idle} server in the latter case is equal to $(1-\underline p(B))\mu N_B^-$. In the second case of $I\leq 0$, $I$ increases if $i$ is equal to the upper boundary of I/W level and increases by one (see \figurename~\ref{fig3}). This occurs when either a \textsc{Busy} server turns \textsc{Idle} with rate $\mu N_B$ or a \textsc{Setup} server becomes \textsc{Idle} with rate $\gamma A^+ K_{\mathrm I}$, thus serving a waiting job. In the latter, a \textsc{Setup} server becomes \textsc{Busy} and thus, to have the same \textsc{Busy} level, $b$ should not be equal to the upper boundary of the current \textsc{Busy} level.

Now we turn to deriving $h_{\circ,-}(S,A^+)$, which captures the transition rate from  $S\!=\!(B,I)$ to $S'=(B,I\!-\!1)$ under $A^+$.~Note that $I$ decreases by one if a job arrives with rate $\lambda$ and $i$ is equal to the lower boundary of I/W level. Then, if $I\!<\!0$, $B$ remains unchanged since the newly arrived job increases the number of waiting jobs. However, if $I\!\geq\!0$, then the new job turns an \textsc{Idle} server \textsc{Busy}; thus, in order to remain in the same \textsc{Busy} level, the number of \textsc{Busy} servers should not be equal to the upper boundary of level $B$. We thus get:
\begin{equation}
	h_{\circ,-}(S,A^+) =	\begin{cases}
								\left( 1- \overline p(B)\right)\lambda \underline u(S,A);  & I \ge 0, \\
								\lambda \underline u(S,A);  & I < 0.
							\end{cases}
	\label{eq23}
\end{equation}

\textbf{Joint \textsc{Busy} and {I/W} Level Transition Rates ($h_{+,+}$, $h_{+,-}$ and $h_{-,+}$).} 
\label{sec4.4.3}
Transition from $S\!=\!(B,I)$ to $S'=(B + 1,I - 1)$ under $A^+$ occurs with rate $\lambda$ only when $I \!\geq\! 0$ and the values of \textsc{Busy} and \textsc{Idle} servers are equal to the upper and lower~boundaries of the corresponding levels, respectively. It thus transpires with a rate:  
\begin{equation}	 
	h_{+,-}(S,A^+) =	\begin{cases}
									\lambda \overline p(B)\underline u(S,A); & I \ge 0, \\
									0; & I<0.
								\end{cases}
	\label{eq24}
\end{equation}
On the other hand, $S\!=\!(B,I)$ transits to $S'=(B-1,I+1)$ under $A^+$ only when $I \!\geq\! 0$ and the number of \textsc{Busy} and \textsc{Idle} servers are equal to the lower and upper boundaries of the corresponding levels, respectively, which further yields:
\begin{equation}		
	h_{-,+}(S,A^+) =	\begin{cases}
									U_B\mu \underline p(B) \overline u(S,A); & I \ge 0, \\
									0; & I<0.
\end{cases}
	\label{eq25}
\end{equation}

Finally, transition from $S\!=\!(B,I)$ to   $S'=(B\!+\!1,I\!+\!1)$ under $A^+$ happens at rate:
\begin{equation}
	h_{+,+}(S,A^+) =	\begin{cases}
									\gamma A^+K_{\mathrm I} \overline p(B) \overline u(S,A); & I < 0, \\
									0; & I \ge 0.
								\end{cases}
	\label{eq26}
\end{equation}

\subsection{The Reward Function $\ragg$}
\label{sec5.3}
Let us define the rate function $\Psiagg$ associated to $\qagg$ as:
\begin{equation}
	\Psiagg\!=\! h_{+,\circ}\!+\!h_{-,\circ}\!+\!h_{\circ,+}\!+\!h_{\circ,-}\!+\!h_{-,+}\!+\!h_{+,-}\!+\!h_{+,+}.
\end{equation}
It is evident that $1/\Psiagg(S,A)$ is the average time spent at state $S$ under action $A$.
Hence, the reward in state $S$ under action $A$~is:
\begin{align*}
    	\ragg\left(S,A\right)
    	&\!=\! \frac{-1}{\Psiagg(S,\!A)}\Big(\! c_{\text{\scriptsize perf}} \lvert\overline I(S,A)^-\rvert + c_{\text{\scriptsize power}} {\overline I}(S,A)^{+\!} + c_{\text{\scriptsize power}}'  K_{\mathrm I}A^{+\!} \Big),
\end{align*}
where $\overline I(S,A)$ is defined in (\ref{eq:I_agg_average}).

\subsection{Solving Multi-level CTMDP}
\label{sec5.4}
Armed with the characterization of $\Magg$, we can derive a similar result to Lemma \ref{lem:CTMDP_time} for $\Magg$. 

\begin{lemma}
\label{lem:CTMDP_time_Magg}
The (per-step) time complexity of solving  the MDP
associated to $\Magg$ using value iteration is $O(L^2(Q/K_\mathrm I+L^2)(Q/K_\mathrm I+L))$.
\end{lemma}

\begin{proof}
The proof follows a similar argument as in the proof of Lemma \ref{lem:CTMDP_time}. Each iteration in value iteration costs $O(|\mathcal S|\sum_{S\in\Sagg}|{\Aagg_S}|)$. Introduce $\Sagg_+:=\{(B,I)\in \Sagg, I>0\}$ and $\Sagg_-:=\{(B,I)\in \Sagg, I\le 0\}$ so that $\Sagg_-\cup \Sagg_+ = \Sagg$. Furthermore,  $\sum_{S\in\Sagg}|{\Aagg_S}| = \sum_{S\in\Sagg_+}|{\Aagg_S}| + \sum_{S\in\Sagg_-}|{\Aagg_S}|$. 
\begin{itemize}
    \item \textit{When $S=(B,I)\in \Sagg_-$:} Recalling that $B \in \{0,\ldots,L-1\}$ and $-\lfloor Q/K_\mathrm I \rfloor \leq I \leq -1$, we have $|\Sagg_-|=L \lfloor Q/K_\mathrm I \rfloor$. Furthermore, $\sum_{s\in \Sagg_-} |\Aagg_S|=2L\lfloor Q/K_\mathrm I \rfloor$ since there are two possible actions in $S$.
    \item \textit{When $S=(B,I)\in \Sagg_+$:} Since $0 \leq I \leq L-1$, $B$ can take different values depending on the value of $I$. But we assume the worst case where $B$ can take all values in $\{0,1,\cdots,L-1\}$ at any value of $I$. Then, we have $|\Sagg_+|=L^2$. Also, is such states, we have $|\Aagg_S|=I+2$.  Consequently, $\sum_{S\in\Sagg_+} |\Aagg_S|=\sum_{I=0}^{L-1}L(I+2)=L^2(L/2+3/2)$. 
\end{itemize}  
Consequently, $|\Sagg| \sum_{S \in \Sagg} |\Aagg_S| = L^2 (L+\lfloor Q/K_\mathrm I \rfloor)(L^2/2+2\lfloor Q/K_\mathrm I \rfloor+3L/2)$, leading to the approximate complexity $O(L^2(L^2+Q/K_\mathrm I)(Q/K_\mathrm I+L))$. This completes the proof.
\end{proof}


\section{Simulation Results and Discussions}
\label{sec6}
In this section, we assess the efficacy of our multi-level CTMDP using numerical experiments. Since our proposed approach assumes known and fixed parameters, it is considered an offline optimization method. As a result, we do not utilize real traces with time-varying rates for performance evaluation in this paper.
This aligns with previous analytical studies on multi-server scenarios with setup time (\cite{Gandhi2010,phung2019delay, Gandhi2013, Tuan2017, Longo2011, Wang2015, Chang2016, Esa2018, Maccio2018}), which evaluate their methods using events generated from distribution functions with known parameters. Moreover, existing traces, to the best of our knowledge, do not provide information about server setup times, which are fundamental components of the system model in this paper. To evaluate~the~performance,~we compare the results with the \textit{staggered threshold} and \textit{bulk setup} policies, using parameters from \cite{Maccio2018}, and the \textit{uniform state-aggregation} method in \cite{Ren2002}. The equivalent (discrete-time) MDPs are solved using linear programming methods for multi-chain MDPs \cite{puterman2014markov} and the Gurobi Java plugin \cite{bixby2011gurobi}. Henceforth, we assume $c_{\text{\scriptsize power}} =1$ and $c'_{\text{\scriptsize power}} =2$, indicating that each \textsc{Setup} server consumes twice the power of an \textsc{Idle} server.
\begin{figure}[!t]
	\centering
	\includegraphics[width=.75\columnwidth]{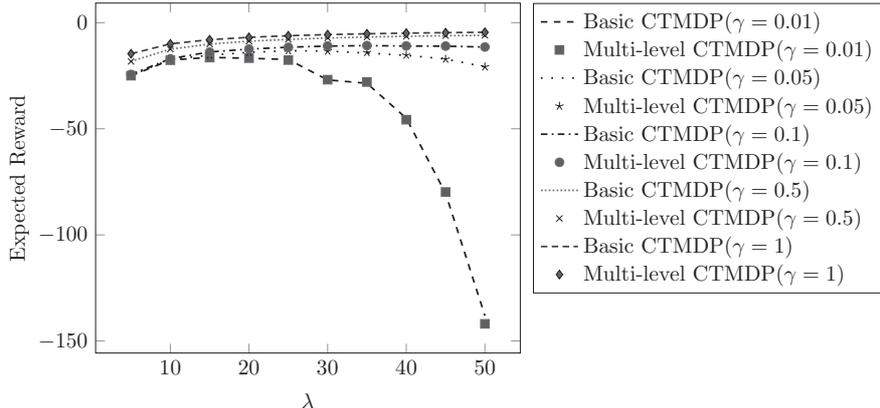}
	\vspace{-1.4em}
	\caption{Validation of the Multi-level CTMDP model $\Magg$ with respect to the basic CTMDP model $M$ for $L\!=\!C \!=\!100$.}
	\label{fig4}
	\vspace{-0.5em}
\end{figure}

\figurename~\ref{fig4} compares the multi-level and basic CTMDPs assuming $Q \!=\!C\!=\! 100$, where we set $L \!=\!C\!=\!100$ (thus, $K_\mathrm B\!=\!K_{\mathrm I}\!=\!1$). The optimal expected average reward is plotted versus the arrival rate $\lambda$, where $\mu\!=\!1$, $c_{\text{\scriptsize perf}}\!=\!50$ are fixed. As the figure shows, the multi-level and basic CTMDPs exhibit exactly the same performance for varying setup and arrival rates.

Two fixed-threshold methods, namely staggered threshold and bulk setup policies \cite{Maccio2018}, have been reported in the literature for power management in multi-server systems with setup times.~We denote them by $\pi_{\text{stag}}$ and $\pi_{\text{bulk}}$, respectively. 
Both policies use a threshold parameter $C_\text{s}$, called `\textit{static \textsc{On}}' servers, which~represents the number of servers that should always be powered on. 
Mathematically, they are defined as \cite{Maccio2018}:
\begin{equation}
	\pi_{\text{bulk}}(b,i) = \begin{cases}
	C_\text{s}\!-\!b\!-\!i^+; & b+\!i^+ \!\le\! C_\text{s}, \\
	\left(C_\text{s}\!-\!b\right)^+\!-\!i; & b\!+\!i \!>\! C_\text{s}, i \!>\! -k, \,\,\\
	C\!-\!b;~~ &  b\!>\!C_\text{s}, i\!\le\!-k.				
	\end{cases}
	\label{31}
\end{equation}
Here, $\pi_{\text{stag}}$ is the same as $\pi_{\text{bulk}}$, except that $\pi_{\text{stag}}(b,i)=|i|$ when $b+i>C_\text{s}$ and $i\!\ge\!0$. In both policies, greater values of $k$ means higher priority of power over delay. Here, we set the threshold $k=1$ to get the highest priority of delay over power. We consider $C_\text{s} \!=\! \rho +\! \sqrt{\rho}$, which is shown in \cite{Maccio2018} to be the optimal value $C_\text{s}$.
\begin{figure*}[!t]
	\centering
	\includegraphics[width=\textwidth]{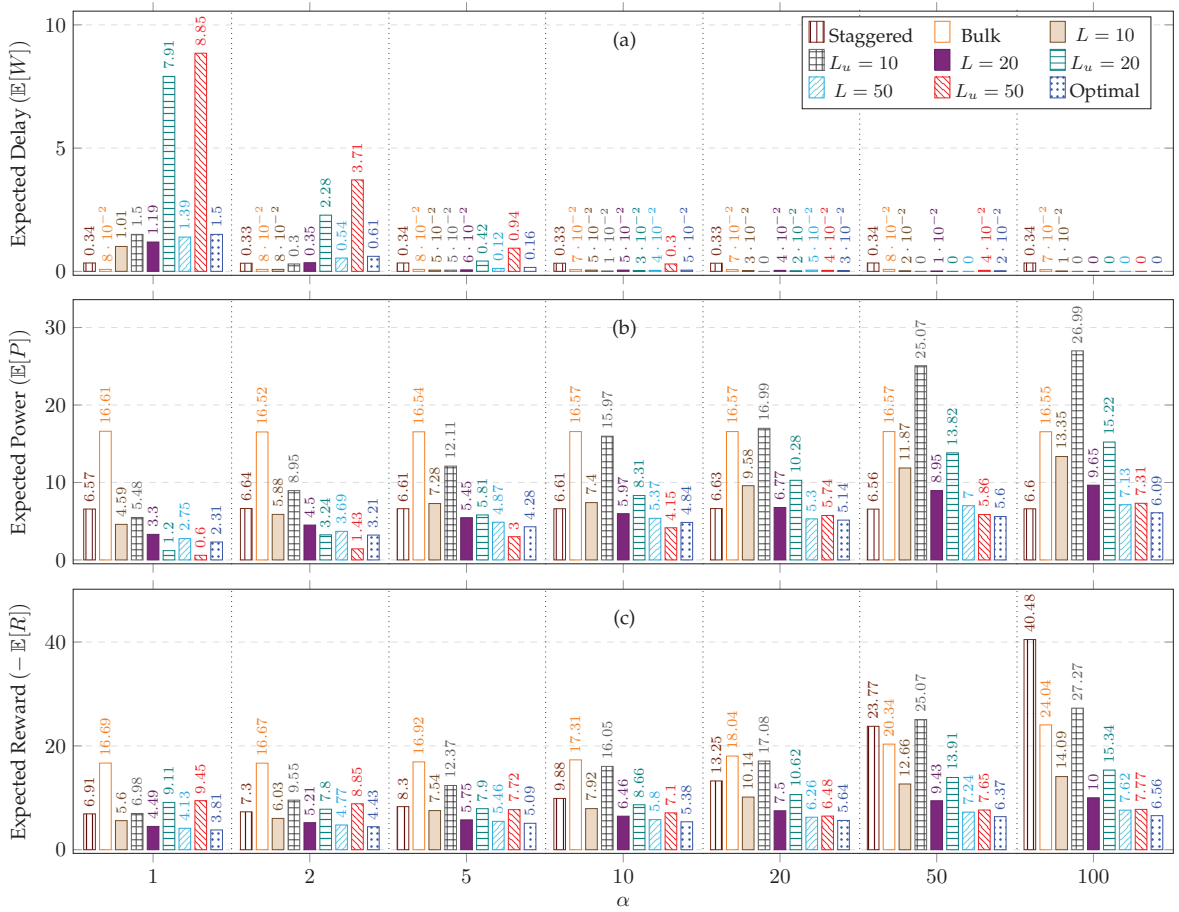}
	\vspace{-0.7em}
	\caption{Expected delay, power, and reward per time unit for varying $c_{\text{\scriptsize perf}}$ values under different policies ($\gamma=2,\lambda=30$).}
	\label{fig5}
\end{figure*} 
\begin{figure*}[!t]
	\centering
	\includegraphics[width=\textwidth]{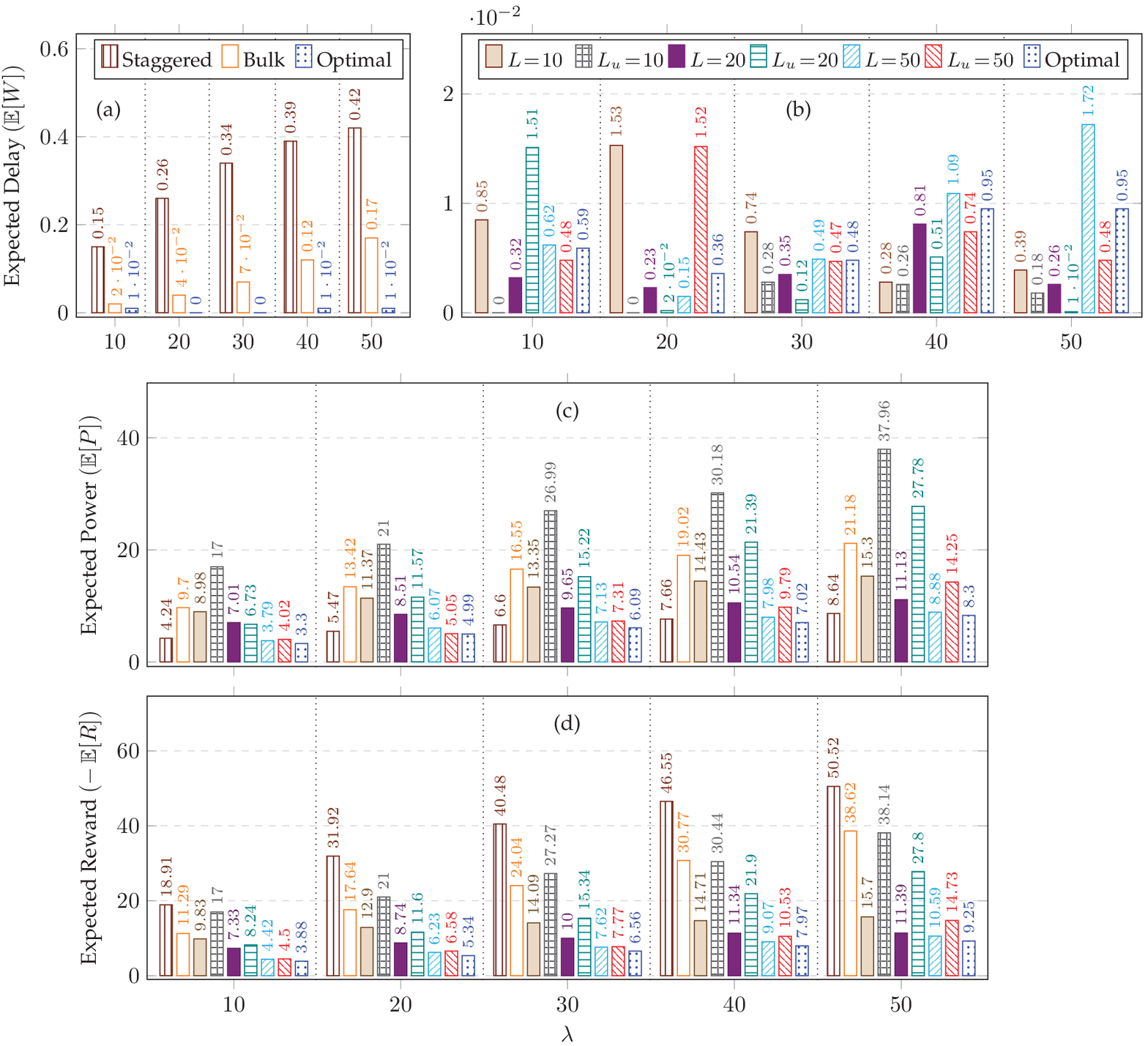}
	\vspace{-0.7em}
	\caption{Expected delay, power, and reward per time unit for varying $\lambda$ values under different policies ($\gamma=2,c_{\text{\scriptsize perf}}=100$).}
	\label{fig6}
\end{figure*}

We also compare our CTMDP model with the \textit{uniform state-aggregation} approach, which derives the reward and transition rate of a meta-state by averaging those of the corresponding aggregated states. Assuming $L_u$ levels in this method, we have $K \!=\! K_\mathrm B \!=\! K_{\mathrm I} \!=\! C/L_u$. Hence, under this method, $A = x$ corresponds to $a=xK$ in $M$, and for $S,S'\in \Sagg$, we have:
\begin{align*}
	&q(S'|S,A)=\frac{1}{K^4}\sum_{w=0}^{K-1}\sum_{x=0}^{K-1}\sum_{y=0}^{K-1}\sum_{z=0}^{K-1}q\big((B'K\!+\!w,I'K\!+\!x)|(BK\!+\!y,IK\!+\!z),AK\big) .
\end{align*}


In what follows, terms $\E[W]$ and $\E[P]$ refer to the average delay and average power of the policies per time unit, respectively. However, in regard to the definition of the reward~in~\eqref{eq4}, the term $\E[W]$ essentially denotes the average number of waiting jobs ($|i^-|$), and $\E[P]$ is the average weighted sum of the number of \textsc{Idle} and \textsc{Setup} servers ($c_{\text{\scriptsize power}}i^+ \!+\! c_{\text{\scriptsize power}}'a^+$). \figurename~\ref{fig5} compares the policies discussed in terms of $\E[W]$, $\E[P]$, and expected reward ($\E[R]$) for values of $c_{\text{\scriptsize perf}}$ taken from $\{1,2,5,10,20,50,100\}$ and $\lambda\!=\!30$, $\mu\!=\!1$, $\gamma\!=\!2$. Since $\pi_{\text{bulk}}$ turns on all \textsc{Off} servers whenever the number of waiting jobs is greater than the given threshold $k$, it prioritizes the delay~as compared to $\pi_{\text{stag}}$ and thus, consumes more power and the jobs receive service with less delay (see \figurename~\ref{fig5}a and \figurename~\ref{fig5}b). Furthermore, for larger $c_{\text{\scriptsize perf}}$ values ($c_{\text{\scriptsize perf}} \!\ge\! 50$), which indicate delay being prioritized over power, $\pi_{\text{bulk}}$ results in higher reward since it prioritizes delay, while for smaller $c_{\text{\scriptsize perf}}$ values, $\pi_{\text{stag}}$ outperforms $\pi_{\text{bulk}}$.  Since both $\pi_{\text{bulk}}$ and $\pi_{\text{stag}}$ are independent of $c_{\text{\scriptsize perf}}$, the power and delay in these methods do not change for different $c_{\text{\scriptsize perf}}$ values. Moreover, it is evident in \figurename~\ref{fig5}c that the absolute value reward of the multi-level CTMDP decreases with $L$ since at larger values of $L$, we have a more accurate model. On the other hand, based on the dimensionality analysis made in Section \ref{sec4}, for $L=50$, $L=20$ and $L=10$, the size of multi-level CTMDP is $32$, $3125$, and $100,000$ times smaller than the optimal CTMDP, respectively. Moreover, even for the smallest $L$ value ($L\!=\!10$) in our experiment, the reward achieved from multi-level CTMDP is better than both $\pi_{\text{bulk}}$ and $\pi_{\text{stag}}$. It should also be noted in \figurename~\ref{fig5}c that the reward of the uniform state-aggregation method is at most equal to that of the multi-level CTMDP. This shows that our method offers a better approximation of the basic CTMDP than the  uniform state-aggregation method.

For varying values of $\lambda \!\in\! \{10,20,30,40,50\}$, $\E[W]$, $\E[P]$, and $\E[R]$ of the different policies are compared in \figurename~\ref{fig6}. The parameters $c_{\text{\scriptsize perf}}\!=\!100$, $\mu\!=\!1$, and $\gamma\!=\!2$ are set to be fixed. As shown in \figurename~\ref{fig6}a and \figurename~\ref{fig6}b, the mean delays of both $\pi_{\text{bulk}}$ and $\pi_{\text{stag}}$ are much larger (about 100 times) than optimal and multi-level approaches. The power and delay of $\pi_{\text{bulk}}$ and $\pi_{\text{stag}}$ increase with $\lambda$ in \figurename~\ref{fig6}b and \figurename~\ref{fig6}c, thus leading to lower rewards as depicted in \figurename~\ref{fig6}d. Indeed, by increasing $\lambda$, the traffic density of the system increases, which results in longer average waiting time. On the other hand, more servers will be in \textsc{Setup} state to deal with higher traffic which leads to more energy consumption too. However, such a monotonic increase in power and delay cannot be observed in CTMDP-based approaches since they minimize the weighted sum of power and performance penalties and thus, improving one component may affect the other component for different $\lambda$ values. Nevertheless, it is~apparent in \figurename~\ref{fig6}d that the reward decreases with $\lambda$ in CTMDP-based approaches.

Finally, \figurename~\ref{fig7} compares the mean delay, power, and reward calculated for different values of $\gamma \!\in\! \{0.1,0.5,1,2,5\}$, where $c_{\text{\scriptsize perf}}=100$, $\mu=1$, and $\lambda=30$. In \figurename~\ref{fig7}a and \figurename~\ref{fig7}b, it can~be observed that for bulk setup and staggered threshold policies, increasing $\gamma$ results in the decrease of the delay, because the setup process finishes faster and the jobs experience less delay. The power also drops with $\gamma$ as shown in \figurename~\ref{fig7}c since decrease in the \textsc{Setup} delay brings about reduction in the number of \textsc{Setup} servers. Similar to \figurename~\ref{fig6}, for CTMDP-based approaches, such a monotonic decrease cannot be observed for delay and power separately, but the increase in reward with respect to $\gamma$ is clearly evident in \figurename~\ref{fig7}d.
\begin{figure*}[!t]
	\centering
	\includegraphics[width=\textwidth]{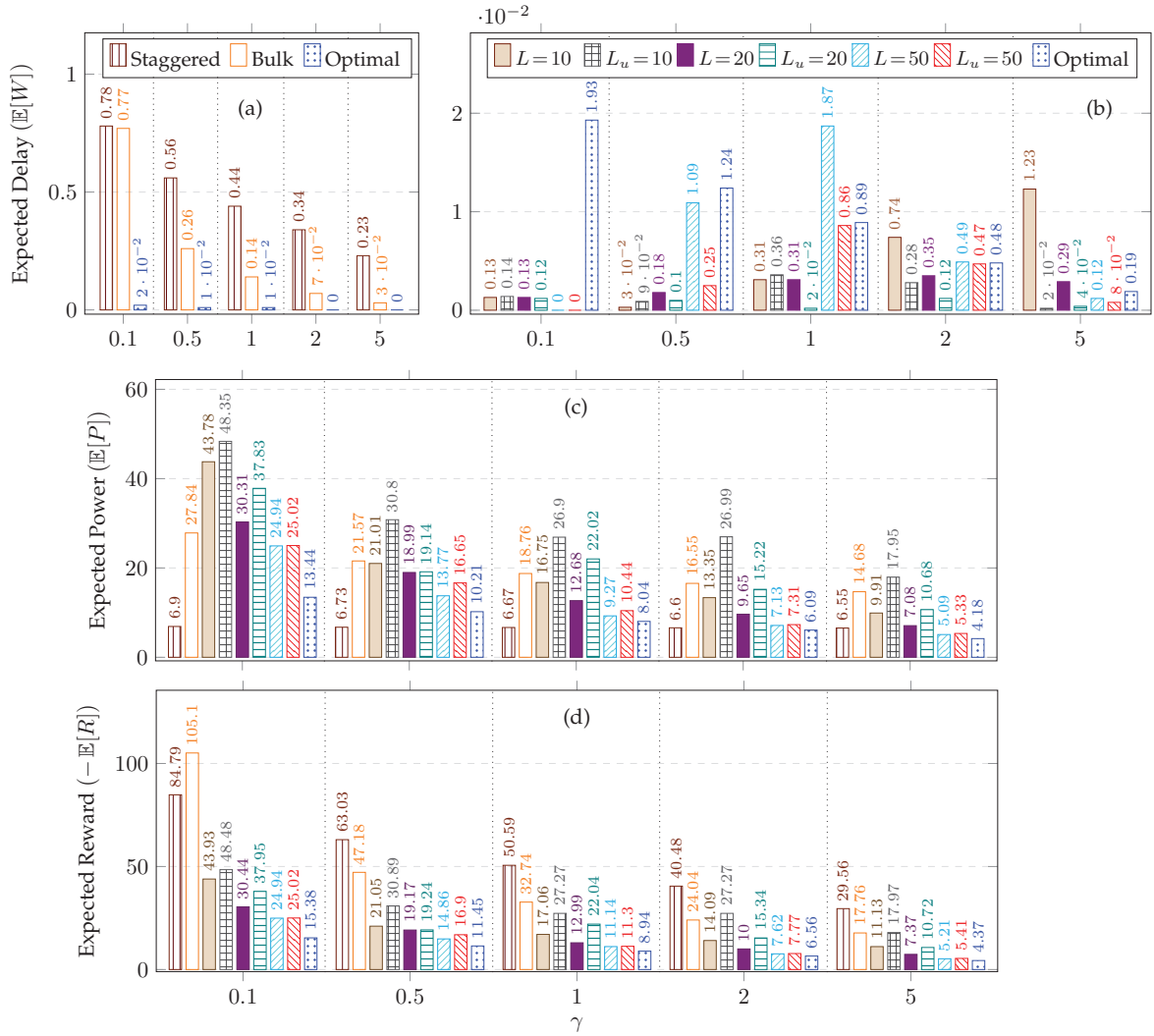}
	\vspace{-0.7em}
	\caption{Expected delay, power, and reward per time unit for varying $\gamma$ values under different policies ($\lambda=30,c_{\text{\scriptsize perf}}=100$).}
	\label{fig7}
\end{figure*}

\section{Conclusion}
\label{sec7}
We have presented a multi-level CTMDP as an approximate model for power management in large scale cloud data centers with setup time. The multi-level CTMDP is derived using a novel state aggregation technique that exploits the intrinsic structure of the model. It is fully characterized under mild assumptions and approximations and is shown to admit a significantly smaller state-action space than the exact model, which makes it a viable solution to remedy the curse of dimensionality in large-scale systems. 
Through numerical simulations, we demonstrated that the resulting power management policies are superior to existing fixed threshold methods.  
As future work, it would be intriguing to extend this model by explicitly incorporating power consumption for physical machines and considering the server power switching cost, commonly referred to as the \textit{wear-and-tear} cost. Another promising research direction is to investigate power management using derived models within an online reinforcement learning setting, such as the approaches presented in \cite{jaksch2010near,bourel2020tightening}, where system parameters are unknown.

\section*{Acknowledgment}
This work was supported by the University of Tehran and the Institute for Research in Fundamental Sciences under grant number CS1399-2-02, and in part by the Social Policy Grant (SPG) funded by Nazarbayev University, Kazakhstan.

\bibliographystyle{unsrt}
\bibliography{myreflist}

\end{document}